\documentclass[acmsmall,nonacm]{acmart}

\AtBeginDocument{%
  }

\usepackage{bm} % bold math fonts
\usepackage{enumitem}
\usepackage{subcaption}
\usepackage{xspace}
\usepackage{algorithm}
\usepackage[noend]{algpseudocode}
\usepackage{tikz}
\usetikzlibrary{calc}
\usetikzlibrary{shapes,decorations.markings,arrows.meta,fit}
\usetikzlibrary{decorations.pathreplacing}
\usetikzlibrary{automata,positioning}
\newcommand{\defeq}{\stackrel{\text{def}}{=}}

\newcommand{\vars}{\text{\sf vars}}
\newcommand{\atoms}{\text{\sf atoms}}
\newcommand{\Dom}{\text{\sf Dom}}

\newcommand{\polylog}{\text{\sf polylog}}

\newcounter{magicrownumbers}

\newcommand{\mylist}{\mathsf{list}}

\newtheorem{thm}{Theorem}[section]
\newtheorem{claim}[thm]{Claim}

\theoremstyle{definition}

\usepackage[colorinlistoftodos]{todonotes}

\newtheoremstyle{cited}%
{.5\baselineskip\@plus.2\baselineskip
    \@minus.2\baselineskip}% (space above)
{.5\baselineskip\@plus.2\baselineskip
    \@minus.2\baselineskip}% (space below)
{\itshape}% (body font)
{\parindent}% (indent amount)
{}% {theorem head font}
{.}% {punctuation after theorem head}
{.5em}% {space after theorem head}
{\textsc{\thmname{#1}} \thmnote{\normalfont#3}}% {theorem head spec}

\theoremstyle{cited}

\newcommand{\out}{\textsf{OUT}}

\definecolor{light-gray}{gray}{0.7.2}
\definecolor{goodgreen}{rgb}{0.1, 0.5, 0.1}
\definecolor{burntorange}{rgb}{0.8, 0.33, 0.0}
\definecolor{lightblue}{RGB}{173, 216, 230} %

\newcommand{\mycomment}[1]{#1}
\newcommand{\nop}[1]{}

\newcommand{\mak}[1]{\mycomment{\todo[inline,color=lightblue]{\textsf{#1} \hfill \textsc{--Mahmoud.}}}}

\newcommand{\dano}[1]{\mycomment{\todo[inline,color=orange]{\textsf{#1} \hfill \textsc{--DanO.}}}}

\newcommand{\change}[1]{#1}

\newcommand{\rai}{RelationalAI}

\newcommand{\pg}{\textsf{PG}\xspace}
\newcommand{\ospg}{\textsf{OSPG}\xspace}

\newcommand{\ov}{\overline}

\begin{document}

\title{Output-Sensitive Evaluation of Regular Path Queries}

%%%%%%%%%%%%%%%%%%%%%%%%%%%%%%%%%%%%%%%%%%%%%%%%%%%%%%%%%%%%%%%%%%%%%%%%%%%%%%%%%%%%%%%%%%%%

\author{Mahmoud Abo Khamis}
\orcid{0000-0003-3894-6494}
\email{mahmoudabo@gmail.com}
\affiliation{
  \institution{\rai{}}
  \department{Query Optimizer}
  %\streetaddress{2120 University Ave}
  \city{Berkeley}
  \state{CA}
  %\postcode{94704}
  \country{USA}
}

\author{Ahmet Kara}
\orcid{0000-0001-8155-8070}
\email{ahmet.kara@oth-regensburg.de}
\affiliation{
  \institution{OTH Regensburg}
  \department{Department of Computer Science \& Mathematics}
  \city{Regensburg}
  \country{Germany}
}

\author{Dan Olteanu}
\orcid{0000-0002-4682-7068}
\email{olteanu@ifi.uzh.ch}
\affiliation{%
    \institution{University of Zurich}
    \department{Department of Informatics}
    %\streetaddress{Andreasstrasse 15}
    \city{Zurich}
    %\postcode{8050}
    \country{Switzerland}
}

\author{Dan Suciu}
\orcid{0000-0002-4144-0868}
\email{suciu@cs.washington.edu}
\affiliation{
  \institution{University of Washington}
  \department{Department of Computer Science \& Engineering}
  %\streetaddress{Box 352350}
  \city{Seattle}
  \state{WA}
  %\postcode{98195-2350}
  \country{USA}
}

\renewcommand{\shortauthors}{Mahmoud Abo Khamis, Ahmet Kara, Dan Olteanu, \& Dan Suciu}
%%%%%%%%%%%%%%%%%%%%%%%%%%%%%%%%%%%%%%%%%%%%%%%%%%%%%%%%%%%%%%%%%%%%%%%%%%%%%%%%%%%%%%%%%%%%

\begin{abstract}
% \mak{Abstract proposals can be found in our google doc:\\\url{https://docs.google.com/document/d/1LcA-50xOgfzO8YxJBDpyMf4vsMEd8niz6StCA7Qa6Xg}}

We study the classical evaluation problem for regular path queries: Given an edge-labeled graph and a regular path query, compute the set of pairs of vertices that are connected by paths that match the query.

The Product Graph (\pg) is the established evaluation approach for regular path queries. \pg first constructs the product automaton of the data graph and the query and then uses breadth-first search to find the accepting states reachable from each initial state in the product automaton. Its data complexity is $O(|V| \cdot |E|)$, where $V$ and $E$ are the sets of vertices and respectively edges in the data graph. This complexity cannot be improved by combinatorial algorithms.

In this paper, we introduce \ospg, an output-sensitive refinement of \pg, whose data complexity is $O(|E|^{3/2} + \min(\out\cdot\sqrt{|E|}, |V|\cdot|E|))$, where $\out$ is the number of distinct vertex pairs in the query output. \ospg's complexity is at most that of \pg and can be asymptotically smaller for small output and sparse input. The improvement of \ospg over \pg is due to the unnecessary time wasted by \pg in the breadth-first search phase, in case a few output pairs are eventually discovered. For queries without  Kleene star, the complexity of \ospg can be further improved to $O(|E| + |E| \cdot \sqrt{\out})$.

\end{abstract}




\maketitle

\section{Introduction}

Regular path queries are an essential construct in graph query languages~\cite{AnglesABHRV17,DBLP:conf/sigmod/DeutschFGHLLLMM22}, including industry efforts such as SPARQL from W3C~\cite{SPARQL:2013}, Cypher from Neo4J~\cite{Cypher:SIGMOD:2018}, GSQL from TigerGraph~\cite{GSQL:2021}, PGQL from Oracle~\cite{PGQL:2021}, and more recently the standards SQL/PGQ and  GQL~\cite{DBLP:conf/sigmod/DeutschFGHLLLMM22,DBLP:conf/icdt/FrancisGGLMMMPR23}. Such queries have also been investigated to a great extent in academia, e.g.,~\cite{AnglesABHRV17,Baeza13,CalvaneseGLV03,Wood12}, since their introduction as a natural declarative formalism for expressing path matching in graphs~\cite{MendelzonW95}.

Consider a data graph $G=(V,E,\Sigma)$, where $V$ is the set of vertices and $E$ is the set of edges that are labeled with symbols from a finite alphabet $\Sigma$. A query $Q$ is a regular expression over $\Sigma$. The semantics of $Q$ is given by the set of vertex pairs $(v, u)$ such that $u$ can be reached in $G$ from $v$ via a path labeled with a word from the language $L(Q)$. Beyond this classical semantics, the recent literature investigates the variants where the query output is the list of all simple paths (no repeated vertices)~\cite{MendelzonW95, MartensT19, BaganBG20}, trails (no repeated edges)~\cite{MartensT19, MartensNP23}, or shortest paths from $v$ to $u$~\cite{MartensT19}. The accompanying decision problem, which checks whether a vertex pair $(v,u)$ is in the query output, also received  attention~\cite{BarceloLLW12,Baeza13}.

The classical semantics can be realized by a simple and effective algorithm, called Product Graph (\pg)~\cite{MendelzonW95,Baeza13,MartensT18}. This is the workhorse of practical implementations in several existing graph database systems, e.g., MillenniumDB~\cite{DBLP:journals/dint/VrgocRAAABHNRR23}, PathFinder~\cite{DBLP:conf/semweb/CalistoFMRV24}, TigerGraph~\cite{DBLP:conf/sigmod/DeutschFGHLLLMM22}, and RelationalAI. \pg first constructs the product graph (product automaton) $P_{G,Q}$ of the data graph $G = (V,E,\Sigma)$ and the non-deterministic finite automaton $M_Q = (V_Q, E_Q)$ representing the query $Q$. This product graph has vertices $V\times V_Q$ and an edge $((v,q),(u,p))$ iff, for some symbol $\sigma\in\Sigma$, $G$ and $M_Q$  have  $\sigma$-labeled edges $(v,\sigma,u)$ and respectively $(q,\sigma,p)$. \pg then uses graph searching techniques like breadth-first search to find the accepting states reachable from each initial state in the product graph $P_{G,Q}$. The two steps of \pg can be intertwined, so that the product graph need not be materialized.

The data complexity\footnote{In this paper, we consider the data complexity, where the query $Q$ is fixed and has constant size. The combined complexity of the algorithms in this paper also has a linear factor in the size of the non-deterministic finite automaton encoding $Q$.} of \pg is polynomial. It takes $O(|E|)$ time to create the product graph.\footnote{In the literature, the data complexity for constructing the product graph is stated as $O(|V|+|E|)$. We consider this to be $O(|E|)$: In case $|V|>|E|$, the vertices without edges do not contribute to the query output and can be ignored without loss of generality. We can create the list of vertices with at least one edge (as hash map with expected constant update time) in one scan of $E$ by upserting the list with the vertices of each edge.} The graph searching step is triggered for each vertex in $G$, hence requiring an overall $O(|V|\cdot |E|)$ data complexity.
Remarkably, this data complexity cannot be improved by combinatorial algorithms: There is no combinatorial\footnote{This term is not defined precisely. A combinatorial algorithm has a running time with a low constant in the $O$-notation and is feasibly implementable. It also does not rely on Strassen-like matrix multiplication: Using fast matrix multiplication, the data complexity of \pg can in fact be improved to $O(|V|^{\omega})$, where $\omega = 2.37$ is the matrix multiplication exponent~\cite{CaselS23}.} algorithm to compute $Q$ with data complexity $O((|V|\cdot |E|)^{1-\epsilon})$ for any $\epsilon >0$, unless the combinatorial Boolean Matrix Multiplication conjecture fails~\cite{CaselS23}.\footnote{The lower bound in~\cite{CaselS23} has $|V|+|E|$ instead of $|E|$, but an inspection of the proof shows that our formulation with just $|E|$ also holds. The combinatorial Boolean Matrix Multiplication  conjecture is as follows: Given two $n\times n$ Boolean matrices $A,B$, there is no combinatorial algorithm that multiplies $A$ and $B$ in $O(n^{3-\epsilon})$ for any $\epsilon > 0$.}  A further conditional lower bound considers both the input and output sizes~\cite{CaselS23}: There is no (combinatorial and even non-combinatorial) algorithm to compute $Q$ with data complexity $O((|E|+\out)^{1-\epsilon})$ for any $\epsilon > 0$, unless the sparse Boolean Matrix Multiplication conjecture fails.\footnote{As in footnote ($4$), the lower bound in~\cite{CaselS23} uses $|V|+|E|$ instead of $|E|$, but our formulation using just 
$|E|$ holds as well. The sparse Boolean Matrix Multiplication  conjecture is as follows: Given two $n\times n$ Boolean matrices $A,B$ represented by their adjacency lists $\{(i,j)\mid A[i,j]=1\}$ and $\{(i,j)\mid B[i,j]=1\}$ of $1$-entries, there is no algorithm that multiplies $A$ and $B$ in time $O(m)$, where $m$ is the total number of 1-entries in the input and output: $m = |\{(i,j)\mid A[i,j]=1\}| + |\{(i,j)\mid B[i,j]=1\}| + |\{(i,j)\mid (A\times B)[i,j]=1\}|$.}

\medskip

In this article, we introduce \ospg, a refinement of \pg that is {\em output-sensitive} in the sense that its  running time depends on the size of the query output. 
The data complexity of \ospg is $O(|E|^{3/2} + \min(\out\cdot\sqrt{|E|}, |V|\cdot|E|))$, where $\out$ 
is the query output size.
\ospg achieves this complexity {\em without knowing a priori the size of the query output}.
\change{Moreover, the \ospg complexity is never higher than the complexity of \pg in any case, and can be asymptotically lower in many common cases.
In particular, $|E| \leq |V|^2$ in any graph, and $|E|\leq O(1)\cdot |V|^2$ in any {\em edge-labeled} graph.\footnote{\change{There can be multiple edges between the same pair of vertices, but with different labels. However, the number of labels is constant in data complexity because the query size is a constant, and we only need to consider labels that occur in the query.}}
Hence, $|V|$ ranges from $O(|E|^{1/2})$ to $O(|E|)$, so from fully dense to very sparse graphs.
Therefore, the $\pg$ complexity of $O(|V|\cdot|E|)$ ranges from $O(|E|^{3/2})$ to $O(|E|^2)$,
thus always subsuming the term $|E|^{3/2}$ in the $\ospg$ complexity.
The $\ospg$ complexity is strictly lower when the input graph is sparse and the query output is small.
This regime includes common scenarios, where the query is selective and the number of edges is far from  maximum. In this regime, the improvement of \ospg over \pg is due to the unnecessary time wasted by \pg in the breadth-first search phase when only a few output pairs are eventually discovered.}

For queries without Kleene star, the complexity of \ospg can be further improved to $O(|E| + |E| \cdot \sqrt{\out})$. For  small output, \ospg improves \pg by a factor up to the number $|V|$ of vertices in the data graph. The larger the  query output is, the closer the runtimes of \ospg and \pg get.

\nop{
\begin{example}
    \dano{Here comes the example we discussed on Dec 3.}
    \mak{I added the detailed example as Example~\ref{ex:one-path} explaining precisely what the $\ospg$ approach does. We can summarize it here.}
\end{example}
}

\medskip

This article is organized as follows. Section~\ref{sec:prelims}  introduces preliminaries on regular path queries and data graphs with labeled edges. Section~\ref{sec:problem}  introduces the evaluation problem for regular path queries.
Section~\ref{sec:algo} introduces the \ospg algorithm. \ospg has two logical steps: The first step reduces the evaluation for arbitrary queries to the evaluation for the specific query $ab^*c$. The second step introduces a degree-aware adaptive evaluation for the latter query $ab^*c$. The reduction in the first step relies on the product graph construction pioneered by \pg.
Section~\ref{sec:complexity} analyzes the time complexity of \ospg. The main result is Theorem~\ref{thm:complexity}, which states that \ospg computes the output of a query in time $O(|E|^{3/2} + \min(\out\cdot\sqrt{|E|}, |V|\cdot|E|))$.
Section~\ref{sec:discussion} contrasts \ospg and \pg on specific queries, to highlight the regime where \pg spends too much unnecessary time relative to \ospg.
Section~\ref{sec:special} discusses the evaluation for two restricted query classes: queries without Kleene star, and the query representing the transitive closure expressed using Kleene star.
Section~\ref{sec:related} overviews related work on regular path query evaluation and output-sensitive query evaluation algorithms. Section~\ref{sec:conclusion} discusses directions for future work.

\section{Preliminaries}
\label{sec:prelims}

A set $\Sigma$ of symbols is called an {\em alphabet}.
The set of all strings over an alphabet $\Sigma$ is denoted by $\Sigma^*$.
The set of all strings over $\Sigma$ of length $k$ is denoted by $\Sigma^k$.
The empty string is denoted by $\epsilon$.
A {\em language} $L$ over an alphabet $\Sigma$ is a subset of $\Sigma^*$.

\begin{definition}[Edge-labeled graph]
    An edge-labeled graph $G$ is a triple $(V, E, \Sigma)$ where $V$ is a finite set of vertices,
    $\Sigma$ is a finite set of edge labels, and $E \subseteq V \times \Sigma \times V$ is a set of
    labeled directed edges. In particular, every edge $e \in E$ is a triple $(v, \sigma, u)$
    denoting an edge from vertex $v$ to vertex $u$ labeled with $\sigma$.
\end{definition}

\begin{definition}[A path in an edge-labeled graph]
    Given an edge-labeled graph $G = (V, E, \Sigma)$ and two vertices $v, u\in V$, a {\em path} $p$ from $v$ to $u$ of length $k$ for some
    natural number $k$ is a sequence of $k+1$ vertices
    $w_0 \defeq v, w_1, \ldots, w_{k} \defeq u$ and a sequence of $k$ edge labels $\sigma_1, \sigma_2, \ldots, \sigma_k$
    such that for every $i \in [k]$, there is an edge $(w_{i-1}, \sigma_i, w_{i}) \in E$.
    The {\em label} of the path $p$, denoted by $\sigma(p)$, is the string $\sigma_1 \sigma_2 \ldots \sigma_k \in \Sigma^k$.
\end{definition}

We use the standard definition of regular expressions constructed using the labels from $\Sigma$ and the  concatenation, union, and Kleene star operators. 
Given a regular expression $Q$ over an alphabet $\Sigma$, we use $L(Q)$ to denote the language defined by $Q$.
A language is called {\em regular} if and only if it can be defined by a regular expression.

\begin{definition}[Kleene-free regular expression]
    A regular expression $Q$ is called {\em Kleene-free} if it only uses the concatenation and union operators.
    In particular, it does not use the Kleene star operator.
    \label{defn:kleene-free}
\end{definition}

We also use the standard definitions of nondeterministic and deterministic finite automata, abbreviated as NFAs and DFAs.
Given an NFA (or a DFA) $M$, we represent $M$ as an edge-labeled graph $M=(V, E, \Sigma)$ where $V$ is the set of states, $\Sigma$ is the alphabet, and $E$ is the set of transitions.
In particular, each labeled edge $(q, \sigma, p)\in E$ represents a transition from state $q$ to state $p$ on input $\sigma$.
We use $L(M)$ to denote the language recognized by $M$.
It is known that a language is regular if and only if it can be recognized by an NFA,
and that every NFA can be converted into an equivalent DFA.
Therefore, given a regular expression $Q$, there exists an NFA (and DFA) $M_Q$ such that $L(Q) = L(M_Q)$.

\begin{definition}[Product Graph, $G_1 \times G_2$]
    Given two edge-labeled graphs $G_1 = (V_1, E_1, \Sigma)$ and $G_2 = (V_2, E_2, \Sigma)$
    over the same set of labels $\Sigma$, the {\em product graph} $G_1 \times G_2$ is a graph
    $G=(V, E)$ (without edge labels) where $V \defeq V_1 \times V_2$ and $E$ is defined as follows:
    \begin{align*}
        E \defeq \{((v_1, v_2), (u_1, u_2)) \mid
            \exists \sigma \in \Sigma : (v_1, \sigma, u_1) \in E_1 \wedge (v_2, \sigma, u_2) \in E_2\}
    \end{align*}
    \label{defn:product-graph}
\end{definition}
In the literature, given a regular path query $Q$ over an input graph $G=(V, E, \Sigma)$, it is standard to consider the product graph $P_{G,Q}$ of $G$ with an NFA $M_Q=(V_Q, E_Q, \Sigma)$ that recognizes the language $L(Q)$.
The initial states of $P_{G,Q}$ are pairs $(v, q)$ where $v$ is a vertex in $G$ and $q$ is an initial state of $M_Q$.
The accepting states of $P_{G,Q}$ are pairs $(v, q)$ where $v$ is a vertex in $G$ and $q$ is an accepting state of $M_Q$.
The goal is to find paths from initial states to accepting states in the product graph.

\begin{definition}[$Q$-reachability and $Q$-degree of a vertex in an edge-labeled graph]
    Let $\Sigma$ be an alphabet, $G = (V, E, \Sigma)$ an edge-labeled graph, and $Q$ a regular path query over $\Sigma$.
    Given a vertex $v\in V$, a vertex $u \in V$ is said to be {\em $Q$-reachable}
    from $v$ if there is a path from $v$ to $u$ whose label is in the language $L(Q)$.
    The set of vertices that are $Q$-reachable from $v$ is denoted by $N_Q(v)$.
    We define the {\em $Q$-degree} of $v$ in $G$, denoted by $\deg_Q(v)$, to be the number of vertices in $N_Q(v)$:
    \begin{align*}
        \deg_Q(v) \defeq |N_Q(v)|.
    \end{align*}
    \label{defn:tau-reachable}
\end{definition}

\section{The RPQ Evaluation Problem}
\label{sec:problem}

In this section, we formally define the RPQ problem that we study in this paper.

\begin{definition}
    A Regular Path Query, RPQ or query for short, is a regular expression $Q$ over some alphabet $\Sigma$.
    The RPQ evaluation problem for $Q$ has the following input and output:
    \begin{itemize}
        \item {\bf Input:} An edge-labeled graph $G = (V, E, \Sigma)$, i.e.~where the edge labels
        come from the alphabet $\Sigma$.
        \item {\bf Output:} The set of all pairs of vertices $(v, u) \in V \times V$
        where $u$ is $Q$-reachable from $v$ (Definition~\ref{defn:tau-reachable}).
    \end{itemize}
    \label{defn:rpq}
\end{definition}

In particular, note that we do {\em not} aim to output the paths from $v$ to $u$ themselves, but rather the pairs $(v,u)$ of vertices that are connected by those paths that match the regular expression $Q$.
The actual number of paths could be exponential in the size of the input graph (and even infinite for graphs with cycles), and hence outputting all of them would be infeasible.

\paragraph{Complexity measures}
We use data complexity. In particular, we assume the query $Q$ is fixed.
Hence, a non-deterministic finite automaton $M_Q=(V_Q, E_Q, \Sigma)$ recognizing the regular expression $Q$ can be assumed to have constant size.
Moreover, the size of the alphabet $\Sigma$ can also be assumed to be constant since
edges with labels $\sigma \in \Sigma$ that do not appear in $Q$ can be dropped from the input graph $G$.
We measure the complexity of the RPQ problem in terms of the following parameters:
\begin{itemize}
    \item {\bf The number of vertices}, $|V|$.
    \item {\bf The number of edges}, $|E|$.
    \item {\bf The output size}, $\out$, which is the number of pairs of vertices $(v, u) \in V \times V$ where $u$ is $Q$-reachable from $v$.
\end{itemize}
In particular, note that the output size $\out$ is {\em not} the number of paths that match the regular expression. Instead, it is just the number of different pairs of endpoints
of those paths.
\section{The Output-Sensitive Product Graph Algorithm}
\label{sec:algo}

We describe below our output-sensitive algorithm for solving RPQs, which we call \ospg.
Given an edge-labeled graph $G = (V, E, \Sigma)$ where the labels come from an alphabet $\Sigma$ and
an RPQ $Q$, which is a regular expression over  $\Sigma$,
our algorithm consists of two main steps:
\begin{itemize}
    \item Step 1: Reduce the RPQ $Q$ over the graph $G = (V, E, \Sigma)$ into the RPQ $a b^* c$ over another graph $G'=(V', E', \Sigma' \defeq \{a, b, c\})$
    such that $|V'|=O(|V|)$ and $|E'|=O(|E|)$.
    \item Step 2: Solve the RPQ $a b^* c$ over the graph $G'$ using the output-sensitive product graph algorithm.
\end{itemize}
Step 1 above is directly based on the construction of the product graph, which is also the first step of the traditional \pg algorithm. Step 2 replaces the graph-searching step of the \pg algorithm with a more efficient algorithm that is output-sensitive.
We explain each step below in more detail.

%%%%%%%%%%%%%%%%%%%%%%%%%%%%%%%%%%%%%%%%%%%%%%%%%%%%%%%%%%%%%%%%%%%%%%%%%%%%%%%%%%%%%%%%%%%%
\subsection{Reduction from any RPQ $Q$ to the RPQ $a b^* c$}
\label{subsec:algo:reduction}
Our reduction from any RPQ $Q$ to the RPQ $a b^* c$ essentially shows that $a b^* c$
is the {\em hardest} RPQ.
Let $M_Q=(V_Q, E_Q, \Sigma)$ be an NFA for $Q$, represented as an edge-labeled graph (cf.\@ Section~\ref{sec:prelims}).
First, we construct the product graph $\ov G \defeq G \times M_Q$.
Let $\ov V$ and $\ov E$ be the set of vertices and edges of $\ov G$, respectively.
Then, we construct an edge-labeled graph $G'=(V',E', \Sigma'\defeq \{a, b, c\})$ as follows:
The set of vertices $V'$ is the same as the set of vertices $\ov V$.
The set of edges $E'$ contains the same edges in $\ov E$ but with the additional label $``b"$.
Moreover, $E'$ contains a self-loop labeled $``a"$ for each vertex corresponding to an initial state
in the NFA $M_Q$, and a self-loop labeled $``c"$ for each vertex corresponding to an accepting state in $M_Q$.
Formally:
\begin{align*}
    V' \quad\defeq\quad & \ov V,\\
    E' \quad\defeq\quad & \{((v, q),``b", (u, p)) \mid ((v, q), (u, p)) \in \ov E\}\quad\cup\\
       &\{((v, q), ``a", (v, q)) \mid \text{$v \in V$ and $q$ is an initial state in the NFA $M_Q$}\}\quad\cup\\
    &\{((v, q), ``c", (v, q)) \mid \text{$v \in V$ and $q$ is an accepting state in the NFA $M_Q$}\}
\end{align*}

Based on the above construction, it is straightforward to see that
for any pair of vertices $(v, u) \in V \times V$, there is a path from $v$ to $u$ in $G$ that matches
the regular expression $Q$ if and only if there exist two states $q, p$ in the NFA $M_Q$
such that there is a path from $(v, q)$ to $(u, p)$ in $G'$
that matches the regular expression $a b^* c$.
Hence, we can compute the output pairs $(v, u)$ to the RPQ $Q$ over $G$
by listing the output pairs $((v, q), (u, p))$ to the RPQ $a b^* c$ over $G'$
and projecting $q$ and $p$ away. Moreover, the number of output pairs
$((v, q), (u, p))$ is at most an $O(1)$-factor larger than the number of output pairs $(v, u)$.
(Recall that in data complexity, the query $Q$ is fixed, hence the number of states of the NFA $M_Q$ is a constant.)

\begin{example}
    \label{ex:reduction}
Suppose we have an edge-labeled graph $G=(V, E, \Sigma=\{d, e, f, g\})$
and consider the RPQ $Q = d^*(e\cdot f + g)^*$.
We demonstrate the above reduction from $Q$ to the RPQ $ab^*c$ over another edge-labeled graph $G'=(V', E', \Sigma'=\{a, b, c\})$.
The regular expression $Q$ can be translated into the DFA $M_Q=(V_Q, E_Q, \Sigma)$ that is depicted in Figure~\ref{fig:ex:reduction}. This DFA has states $\{q_0, q_1, q_2\}$, an initial state $q_0$, accepting states $\{q_0, q_2\}$, and the following transitions, given as labeled edges:
\begin{align*}
    E_Q =
    \{(q_0, ``d", q_0),
    (q_0, ``e", q_1),
    (q_0, ``g", q_2),
    (q_1, ``f", q_2),
    (q_2, ``e", q_1),
    (q_2, ``g", q_2)\}
\end{align*}
The product graph $\ov G$ has the set of vertices $\ov V = V \times \{q_0, q_1, q_2\}$
and the following set of edges:
\begin{align*}
    \ov E \quad=\quad
    &\{((v, q_0), (u, q_0))\mid (v, ``d", u) \in E\} &\cup&\quad
    \{((v, q_0), (u, q_1))\mid (v, ``e", u) \in E\} &\cup\\
    &\{((v, q_0), (u, q_2))\mid (v, ``g", u) \in E\} &\cup&\quad
    \{((v, q_1), (u, q_2))\mid (v, ``f", u) \in E\} &\cup\\
    &\{((v, q_2), (u, q_1))\mid (v, ``e", u) \in E\} &\cup&\quad
    \{((v, q_2), (u, q_2))\mid (v, ``g", u) \in E\}
\end{align*}
Finally, the new edge-labeled graph $G'=(V', E', \Sigma')$ has the set of vertices $V' = \ov V$ and the following set of labeled edges:
\begin{align*}
    E' =
    &\{((v, q_0), ``b", (u, q_0))\mid (v, ``d", u) \in E\} &\cup&\quad
    \{((v, q_0), ``b", (u, q_1))\mid (v, ``e", u) \in E\} &\cup\\
    &\{((v, q_0), ``b", (u, q_2))\mid (v, ``g", u) \in E\} &\cup&\quad
    \{((v, q_1), ``b", (u, q_2))\mid (v, ``f", u) \in E\} &\cup\\
    &\{((v, q_2), ``b", (u, q_1))\mid (v, ``e", u) \in E\} &\cup&\quad
    \{((v, q_2), ``b", (u, q_2))\mid (v, ``g", u) \in E\} &\cup\\
    &\{((v, q_0), ``a", (v, q_0))\mid v \in V\} &\cup&\quad
    \{((v, q_0), ``c", (v, q_0))\mid v \in V\} &\cup\\
    &\{((v, q_2), ``c", (v, q_2))\mid v \in V\}
\end{align*}
And now we can verify that there is a path in $G'$ from a vertex $(v, q)$ to another $(u, p)$
that matches $a b^* c$ if and only if there is a path in $G$ from $v$ to $u$ that matches $d^*(e\cdot f + g)^*$.
\end{example}

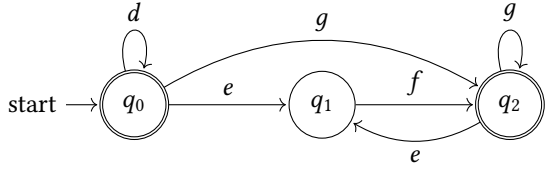
\begin{figure}
\begin{tikzpicture}[shorten >=1pt,node distance=2.5cm,on grid,auto]
    % States
    \node[state, initial, accepting] (q0) {$q_0$};
    \node[state, right=of q0] (q1) {$q_1$};
    \node[state, accepting, right=of q1] (q2) {$q_2$};
  
    % Transitions
    \path[->]
        (q0) edge[loop above] node {$d$} (q0)
             edge node {$e$} (q1)
        (q0) edge[bend left] node {$g$} (q2)
        (q1) edge node {$f$} (q2)
        (q2) edge[bend left] node {$e$} (q1)
        (q2) edge[loop above] node {$g$} (q2);
\end{tikzpicture}
\caption{The DFA for the RPQ $Q = d^*(e\cdot f + g)^*$ from Example~\ref{ex:reduction}.}
\Description{The DFA for the RPQ $Q = d^*(e\cdot f + g)^*$ from Example~\ref{ex:reduction}.}
\label{fig:ex:reduction}
\end{figure}

%%%%%%%%%%%%%%%%%%%%%%%%%%%%%%%%%%%%%%%%%%%%%%%%%%%%%%%%%%%%%%%%%%%%%%%%%%%%%%%%%%%%%%%%%%%%
\subsection{Solving the RPQ $a b^* c$}
\label{subsec:algo:ab*c}
Our output-sensitive algorithm for the RPQ $a b^* c$ is given in Algorithm~\ref{algo:ab*c}.
The input is an edge-labeled graph $G=(V, E, \Sigma=\{a, b, c\})$ and the output is a list of pairs of vertices $(v, u)$
where $u$ is $(ab^*c)$-reachable from $v$.

The first step of the algorithm is to compute a binary relation $R(X, Y)$ that stores pairs of vertices
$(x, y)$ where $y$ is $(b^*c)$-reachable from $x$, as defined by Eq.~\eqref{eq:R:base} and~\eqref{eq:R:recursive} in the algorithm.
However, for every vertex $x$, we only store in $R$ at most $\Delta\defeq \sqrt{|E|}+1$ distinct $y$ values.
If $x$ has more than $\Delta$ distinct $y$ values that are $(b^*c)$-reachable from $x$, we can store any $\Delta$ of them and ignore the rest. By construction, the size of $R$ cannot exceed $|V|\cdot \Delta$.
We will see in Section~\ref{sec:complexity} that the computation of $R$ takes time $O(|E|^{3/2})$.

Given a vertex $x$, we define $\deg_R(x)$ to be the number of distinct $y$ values satisfying $R(x, y)$.
In the second step of the algorithm, we partition vertices $x \in V$ based on $\deg_R(x)$.
In particular, a vertex $x$ is called {\em light} if $\deg_R(x) \leq \sqrt{|E|}$ and {\em heavy} otherwise.
Note that by definition of $R$ and Definition~\ref{defn:tau-reachable},
this is equivalent to saying that a vertex $x$ is light if its $(b^*c)$-degree, $\deg_{b^*c}(x)$, is at most $\sqrt{|E|}$
and heavy otherwise.
We define $R_\ell$ to be the set of pairs of vertices $(x, y) \in V\times V$ where
$x$ is light and $y$ is $(b^* c)$-reachable from $x$, as shown in Eq.~\eqref{eq:R_ell}.
We also define $R_h$ in Eq.~\eqref{eq:R_h} to be the set of heavy vertices $x \in V$.
Given $R$, both $R_\ell$ and $R_h$ are straightforward to compute.

Finally, the algorithm evaluates the Datalog program given by Eq.~\eqref{eq:Q_ell}--\eqref{eq:Q_h}, which defines two binary relations $Q_\ell$ and $Q_h$, and the algorithm returns their union as the final output to the RPQ.
We show below that $Q_\ell$ contains all pairs of vertices $(v, u)$ where $u$ is $(ab^*c)$-reachable from $v$ through a path that only visits light vertices.
In contrast, $Q_h$ contains all pairs of vertices $(v, u)$ where $u$ is $(ab^*c)$-reachable from $v$ through a path that visits at least one heavy vertex.
\change{Therefore, the algorithm computes the correct output to the RPQ $ab^*c$.}
\change{The Datalog program contains a recursive definition for a binary relation $T$,
given by Eq.~\eqref{eq:T:base} and~\eqref{eq:T:recursive}.
We evaluate $T$ using {\em semi-na\"ive evaluation}.
In particular, at every iteration of the recursion, we only examine the {\em new} pairs 
of vertices $(x, z) \in T$ that were just added to $T$ in the previous iteration, and go through vertices $y$ that are $(b)$-reachable from $z$. We add the pairs $(x, y)$ to $T$ if they do not already exist. Only those {\em newly added} pairs to $T$ will be examined in the next iteration, and so on.
We show in Section~\ref{sec:complexity} that using this evaluation strategy, $Q_h$ can be computed in time
$O(\min(\out\cdot\sqrt{|E|}, |V|\cdot|E|))$.
In contrast, $Q_\ell$ takes time $O(|E|^{3/2})$.
}

\begin{algorithm}[th!]
    \caption{The $\ospg$ algorithm for the RPQ $a b^* c$}
    \begin{algorithmic}[1]
        \Statex{\textbf{Input}}: An edge-labeled graph $G=(V,E,\Sigma =\{a, b, c\})$.
        \Statex{\textbf{Output}}: The list of pairs $(v, u) \in V \times V$ where $u$ is $(ab^*c)$-reachable from $v$.
        \Statex{}
        \State Compute $R(X, Y)$ below while maintaining for every $X$ value, at most $\Delta\defeq\sqrt{|E|} + 1$ distinct $Y$ values:\label{algo:ab*c:step1}
        \Statex{\Comment{See Section~\ref{subsec:algo:ab*c}.}}
        \begin{align}
            R(X, Y) &= E(X, ``c", Y)\label{eq:R:base}\\
            R(X, Y) &= E(X, ``b", Z) \wedge R(Z, Y)\label{eq:R:recursive}
        \end{align}
        \State Compute $R_\ell$ and $R_h$, where $\deg_R(X)$ denotes the number of distinct $Y$ values satisfying $R(X, Y)$:\label{algo:ab*c:step2}
        \begin{align}
            R_\ell(X, Y) &= R(X, Y) \wedge \deg_R(X) \leq \sqrt{|E|}\label{eq:R_ell}\\
            R_h(X) &= R(X, Y) \wedge\deg_R(X) = \sqrt{|E|}+1\label{eq:R_h}
        \end{align}
        \State \change{Compute $Q_\ell$ and $Q_h$. For $Q_h$, use {\em semi-na\"ive evaluation}
        to compute $T$:} \label{algo:ab*c:step3} \Comment{\change{See Section~\ref{subsec:algo:ab*c}.}}
        \begin{align}
            Q_\ell(X, Y) &= E(X, ``a", Z) \land R_\ell(Z, Y)\label{eq:Q_ell}\\
            T(X, Y) &= E(X, ``a", Y) \land R_h(Y) \label{eq:T:base}\\
            T(X, Y) &= T(X, Z) \land E(Z, ``b", Y)\label{eq:T:recursive}\\
            Q_h(X, Y) &= T(X, Z) \land E(Z, ``c", Y)\label{eq:Q_h}
        \end{align}
        \State Return $Q_\ell \cup Q_h$ \label{algo:ab*c:step4}
    \end{algorithmic}
    \label{algo:ab*c}
\end{algorithm}
Section~\ref{sec:discussion} demonstrates the above algorithm on a couple of examples
and shows how it outperforms the traditional product graph algorithm.
We prove below that our algorithm is correct.
\begin{theorem}
    Algorithm~\ref{algo:ab*c} computes the correct output to the RPQ $a b^* c$.
    \label{thm:correctness}
\end{theorem}
\begin{proof}
    By definition of $Q_\ell$ and $Q_h$, both are subsets of the output of the RPQ $a b^* c$.
    We prove below that they are also a superset of the output.
    Consider a pair of vertices $(v, u)$ in the output, i.e.~where $u$ is $(a b^* c)$-reachable from $v$.
    Consider a path with label in $(ab^* c)$ consisting of the sequence of vertices $v, w_1, \ldots, w_k, u$ for some $k\geq 1$.
    In particular, the edge set $E$ contains the labeled edges $(v, ``a", w_1)$, $(w_1, ``b", w_2)$, $\ldots$,
    $(w_{k-1}, ``b", w_k)$, $(w_k, ``c", u)$. We recognize two cases:
    \begin{itemize}
        \item Case 1: All vertices $w_1, \ldots, w_{k}$ are light. In this case, $R_\ell(w_1, u)$ holds. Hence by Rule~\eqref{eq:Q_ell}, we have $Q_\ell(v, u)$.
        \item Case 2: There exists some $i \in [k]$ where $w_i$ is heavy.
        In this case, all vertices $w_j$ for $j < i$ are heavy as well.
        This is because for each vertex $w_j$ where $j < i$, we have
        $N_{b^*c}(w_j) \supseteq N_{b^*c}(w_i)$. (Recall Definition~\ref{defn:tau-reachable}.)
        Hence, by Rule~\eqref{eq:T:base}, we have $T(v, w_1)$.
        Moreover, by repeatedly applying Rule~\eqref{eq:T:recursive}, we have $T(v, w_2)$, $\ldots$, $T(v, w_k)$.
        Finally, by Rule~\eqref{eq:Q_h}, we have $Q_h(v, u)$.
    \end{itemize}
\end{proof}
\section{Complexity Analysis}
\label{sec:complexity}
In this section, we prove the following theorem about the time complexity of the \ospg algorithm.

\begin{theorem}
    For any RPQ $Q$ over an edge-labeled graph $G = (V, E, \Sigma)$, the \ospg algorithm computes the output of $Q$ in time $O(|E|^{3/2} + \min(\out\cdot\sqrt{|E|}, |V|\cdot|E|))$.
    \label{thm:complexity}
\end{theorem}
\begin{proof}
Thanks to the reduction in Section~\ref{subsec:algo:reduction}, we only need to analyze the runtime
of Algorithm~\ref{algo:ab*c} for the special RPQ $a b^* c$.
We start with analyzing the runtime needed for the first step of Algorithm~\ref{algo:ab*c}, which computes the relation $R(X, Y)$.
Recall that for each vertex $x\in V$, $R$ contains at most $\Delta\defeq\sqrt{|E|}+1$ different vertices $y$ that are $(b^*c)$-reachable from $x$.
Therefore, $|R| \leq |V|\cdot \Delta$.
In order to compute $R$, we maintain, for every vertex $x$, a list $\mylist(x)$ of up to $\Delta$ different vertices $y$ that are $(b^*c)$-reachable from $x$, as follows:
\begin{itemize}
    \item For every vertex $x \in V$, we initialize $\mylist(x)$ to contain up to
    $\Delta$ different vertices $y$ that are $(c)$-reachable from $x$. This can be done in time $O(|E|)$.
    \item We repeatedly expand the lists $\mylist(x)$ for all $x \in V$ as follows,
    until no new vertex is added to any list:
    Whenever a {\em new} vertex $y$ is added to $\mylist(x)$ for some vertex $x$, we go through all vertices
    $w$ such that there is an edge $(w, ``b", x) \in E$ and check whether $y$ is already
    in $\mylist(w)$. If not and $|\mylist(w)| < \Delta$, then we add $y$ to $\mylist(w)$.
    During this entire expansion process, each edge $(w, ``b", x)$ will be checked
    against every vertex $y$ in $\mylist(x)$ only once.
    Hence, the total runtime is $O(|E|\cdot \Delta)=O(|E|^{3/2})$.
\end{itemize}
Once $R$ is computed, the relations $R_\ell$ and $R_h$ in the second step of Algorithm~\ref{algo:ab*c} can be computed in linear time in the size of $R$.
    
    Finally, we analyze the runtime of the third step of the algorithm, which computes $Q_\ell$ and $Q_h$.
    The relation $Q_\ell$ can be straightforwardly computed in time $O(|E|^{3/2})$. This is because in Rule~\eqref{eq:Q_ell},
    the definition of $R_\ell$
    implies that for every $Z$-value, we have at most $\sqrt{|E|}$ $Y$-values.
    We are left to analyze the runtime of computing $Q_h$.
    \begin{claim}
        $Q_h$ can be computed in time $O(\min(\out\cdot\sqrt{|E|}, |V|\cdot|E|))$.
        \label{clm:Qh-runtime}
    \end{claim}
    \begin{proof}[Proof of Claim~\ref{clm:Qh-runtime}]
        Consider the set of vertices $S_h$ that is defined as follows:
        \begin{align*}
            S_h(X) = E(X, ``a", Y) \land R_h(Y)
        \end{align*}
        The size of $S_h$ cannot exceed $\min\left(\frac{\out}{\sqrt{|E|}}, |V|\right)$. This is because for every
        vertex $x \in S_h$, there are more than $\sqrt{|E|}$ different vertices $y$ such that
        the pair $(x, y)$ is in the output of the RPQ $a b^* c$, and the total number of such
        pairs $(x, y)$ is $\out$ by definition.
        By Rule~\eqref{eq:T:base} and then inductively by Rule~\eqref{eq:T:recursive},
        we have $\pi_X T(X, Y) \subseteq S_h$.
        By Rule~\eqref{eq:Q_h}, we have $\pi_X Q_h(X, Y) \subseteq \pi_X T(X, Z) \subseteq S_h$.
        \change{For a {\em fixed} value $x \in S_h$ of the variable $X$,
        the semi-na\"ive evaluation of Rules \eqref{eq:T:base} and \eqref{eq:T:recursive}
        takes $O(|E|)$ time because it corresponds to a single-source traversal of the input graph where the source is $x$.
        In particular, during this single-source traversal,
        every vertex $z$ that is $(ab^*)$-reachable from $x$ will be added to $T$ only once,
        and the outgoing edges from $z$ will be examined only once, thus leading to a total of
        $O(|E|)$ time.
        Rule \eqref{eq:Q_h} also takes $O(|E|)$ time for a fixed value $x$ of $X$ because the other two
        variables $Z$ and $Y$ form an edge, $E(Z, ``c", Y)$.
        Hence, the total time over all possible values $x \in S_h$ of the variable $X$
        is
        $O\left(\min\left(\frac{\out}{\sqrt{|E|}}, |V|\right)\cdot |E|\right) = O(\min(\out\cdot\sqrt{|E|}, |V|\cdot|E|))$, as desired.}
    \end{proof}
\end{proof}

\section{\change{When is \ospg strictly better than \pg?}}
\label{sec:discussion}

\change{As mentioned in the introduction, the $\ospg$ complexity of $O(|E|^{3/2} + \min(\out\cdot\sqrt{|E|}, |V|\cdot|E|))$ is never higher than the $\pg$ complexity of $O(|V|\cdot|E|)$ in any case.
This is because every edge-labeled graph satisfies $|E| \leq O(1)\cdot|V|^2$.
Hence, $|V|$ ranges from $O(|E|^{1/2})$ to $O(|E|)$, and the $\pg$ complexity ranges from
$O(|E|^{3/2})$ to $O(|E|^2)$.
Moreover, the $\ospg$ complexity is strictly lower than $\pg$ when the input graph is sparse and the query output is small.
This holds specifically for $\out = O((|V|\cdot\sqrt{|E|})^{1-\alpha})$ and $|E|=O(|V|^{2-\beta})$ for any $\alpha \in (0,1]$ and $\beta \in (0,2)$. This regime includes common scenarios, where the query is selective and the number of edges is far from maximum. In this regime, the improvement of \ospg over \pg is due to the unnecessary time used by \pg in the breadth-first search phase when only a few output pairs are eventually discovered. We next exemplify this gap between \ospg and \pg in two cases.}

\begin{example}
    \label{ex:one-path}
    Consider an input edge-labeled graph representing a path of length $N$ where all edges have a label $``b"$.
    In particular, suppose that the vertices are $\{1, 2, \ldots, N\}$ and the edges are
    $\{(1, ``b", 2), (2, ``b", 3), \ldots, (N-1, ``b", N)\}$.
    Consider the RPQ $b^* c$.
    On this instance, the output is empty since there is no path with label $b^* c$.
    The NFA for $b^* c$ has an initial state $q_0$, an accepting state $q_1$, and transitions
    $(q_0, ``b", q_0)$ and $(q_0, ``c", q_1)$.
    The product graph has vertices $\{(1, q_0), \ldots, (N, q_0)\} \cup \{(1, q_1), \ldots, (N, q_1)\}$
    and edges $\{((1, q_0), (2, q_0)), \ldots, ((N-1, q_0), (N, q_0))\}$.
    \begin{itemize}
    \item The traditional $\pg$ algorithm takes time $\Omega(N^2)$ on this instance.
    This is because for every one of the vertices $\{(1, q_0), \ldots, (N, q_0)\}$ that correspond to the initial state, it will run a BFS to discover a path to all subsequent vertices. Yet, it will never discover a way to extend any of these paths to reach any of the vertices $\{(1, q_1), \ldots, (N, q_1)\}$ that correspond to the accepting state.
    \item In contrast, our $\ospg$ algorithm solves this query in $O(1)$ time.
    In particular, consider the relation $R(X, Y)$ that is defined in the first step of Algorithm~\ref{algo:ab*c}. This relation lists for every $X$-value, up to $\sqrt{N}+1$ different $Y$-values that are $(b^*c)$-reachable from $X$.
    In this example, this relation is empty, hence, subsequent relations $R_\ell, R_h, Q_\ell$, and $Q_h$ that are computed in the remaining steps of Algorithm~\ref{algo:ab*c} are also empty.
    \end{itemize}
\end{example}

In the above example, it is possible to argue that the only reason $\ospg$ outperforms $\pg$
is because $\pg$ starts BFS from the beginning of the regular expression $b^*c$.
In contrast, if we were to modify $\pg$ to start BFS from the end of the regular expression and work backwards,
then every BFS that $\pg$ performs would terminate in constant time, hence $\pg$ would solve this query in $O(|V|) = O(N)$ time.
However, in the next example, we will show that this is not really the main issue with $\pg$:
Even if we extend the $\pg$ algorithm to start BFS simultaneously from {\em both} the beginning {\em and} the end of the regular expression, it would still perform poorly compared to $\ospg$ on some instances.

\begin{example}
    \label{ex:two-cycles}
    Now consider an input edge-labeled graph consisting of two disjoint cycles of length $N$ each:
    \begin{itemize}
        \item In the first cycle, every pair of consecutive vertices is connected by two edges: one with label $``a"$
        and the other with label $``b"$.
        \item In the second cycle, every pair of consecutive vertices is connected by two edges: one with label $``b"$
        and the other with label $``c"$.
    \end{itemize}
    \begin{figure}
    \begin{tikzpicture}[scale = 0.7, yscale=0.6, every node/.style={scale=.8, inner sep = 2}]
        \begin{scope}
            % Define the number of nodes
            \def\numnodes{6}
            % Define the radius of the cycle
            \def\radius{3cm}

            \node at (0, 0) {First cycle};
            
            % Draw nodes in a circular arrangement
            \foreach \i in {1,...,\numnodes} {
                \node[circle, draw, fill=white] (node\i) at ({180-360/\numnodes * (\i-1)}:\radius) {$\i$};
            }
            
            % Draw double arrows between consecutive nodes
            \foreach \i in {1,...,\numnodes} {
                % Compute the next node index (wraps around to 1 after the last node)
                \pgfmathtruncatemacro{\j}{mod(\i, \numnodes) + 1}
                
                % Draw two arrows: one with a left bend and one with a right bend
                \draw[->] (node\i) edge[bend left=15, gray] node[midway, black]{$a$} (node\j); % First arrow
                \draw[->] (node\i) edge[bend right=15, gray] node[midway, black]{$b$} (node\j); % Second arrow
            }
        \end{scope}

        \begin{scope}[shift = {(8cm, 0)}]
            % Define the number of nodes
            \def\numnodes{6}
            % Define the radius of the cycle
            \def\radius{3cm}

            \node at (0, 0) {Second cycle};
            
            % Draw nodes in a circular arrangement
            \foreach \i in {1,...,\numnodes} {
                \node[circle, draw, fill=white, inner sep = 1] (node\i) at ({180-360/\numnodes * (\i-1)}:\radius) {$\i'$};
            }
            
            % Draw double arrows between consecutive nodes
            \foreach \i in {1,...,\numnodes} {
                % Compute the next node index (wraps around to 1 after the last node)
                \pgfmathtruncatemacro{\j}{mod(\i, \numnodes) + 1}
                
                % Draw two arrows: one with a left bend and one with a right bend
                \draw[->] (node\i) edge[bend left=15, gray] node[midway, black]{$b$} (node\j); % First arrow
                \draw[->] (node\i) edge[bend right=15, gray] node[midway, black]{$c$} (node\j); % Second arrow
            }
        \end{scope}
    \end{tikzpicture}
    \caption{The input graph for Example~\ref{ex:two-cycles} for $N=6$.}
    % \dano{Can we also show the two NFAs and then the product graphs for \pg and then the graphs G' for \ospg?}
    % \mak{Done}
    \Description{The input graph for Example~\ref{ex:two-cycles}.}
    \label{fig:two-cycles}
    \end{figure}
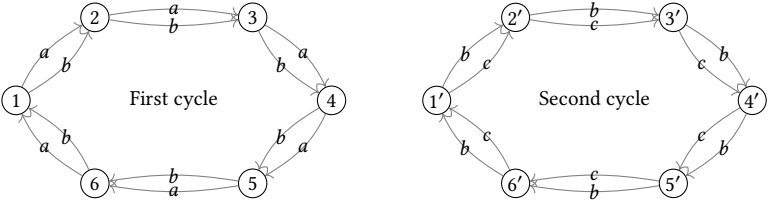
    Figure~\ref{fig:two-cycles} depicts this graph for $N=6$.
    Suppose that the RPQ is $a b^* c$. The output is empty.
    The NFA has states $\{q_0, q_1, q_2\}$ where $q_0$ is an initial state, $q_2$ is an accepting state, and the transitions are $\{(q_0, ``a", q_1), (q_1, ``b", q_1), (q_1, ``c", q_2)\}$.
    The product graph has edges $((v, q_0), (u, q_1))$
    and $((v, q_1),$ $(u, q_1))$ for every pair of consecutive vertices $(v, u)$ in the first cycle.
    It also has edges $((v, q_1), (u, q_1))$ and $((v, q_1), (u, q_2))$ for every pair of consecutive vertices $(v, u)$ in the second cycle.
    The goal is to find a path from $(v, q_0)$ to $(u, q_2)$ for every pair of vertices $(v, u)$.
    \begin{itemize}
        \item Even if we run BFS from both the beginning and the end, the $\pg$ approach still takes time
        $\Omega(N^2)$. This is because neither BFS can immediately discover that the output is empty.
        In particular, forward BFS will discover a path from $(v, q_0)$ to $(u, q_1)$
        for every pair of vertices $v$ and $u$ in the first cycle.
        Similarly, backward BFS will discover a path
        $(v, q_1)$ to $(u, q_2)$ for every pair of vertices $v$ and $u$ in the second cycle.
        \item In contrast, the $\ospg$ algorithm solves this query in $O(N^{3/2})$ time.
        Since the regular expression is $a b^*c$, we can use Algorithm~\ref{algo:ab*c} directly.
        In particular, $\Delta$ defined in the first step of Algorithm~\ref{algo:ab*c} will be $O(\sqrt{N})$.
        Every vertex in the second cycle will have a $(b^*c)$-degree of $O(N)$ and thus will be heavy.
        In contrast, vertices in the first cycle will have a $(b^*c)$-degree of $0$ and won't appear at all in the relation $R(X, Y)$: These are neither heavy nor light. Overall, $R$ will have size $O(N^{3/2})$ and can be computed in the same time.
        The relation $R_\ell$ from Eq.~\eqref{eq:R_ell} will be empty, while $R_h$ from Eq.~\eqref{eq:R_h} will contain every vertex in the second cycle and will have size $O(N)$.
        Since $R_\ell$ is empty, $Q_\ell$ from Eq.~\eqref{eq:Q_ell} will be empty as well.
        Although $R_h$ is not empty, $T$ will be empty because the join $E(X, ``a", Y) \land R_h(Y)$ from Eq.~\eqref{eq:T:base} is empty. As a result, $Q_h$ from Eq.~\eqref{eq:Q_h} is empty.
    \end{itemize}
\end{example}
\section{Special classes of RPQs}
\label{sec:special}

In this section, we study two interesting special classes of RPQs showing that they admit faster or simpler algorithms than the general \ospg algorithm from Section~\ref{sec:algo}:
\begin{itemize}
    \item {\em Kleene-free RPQs:} These are RPQs that do not involve the Kleene star operator;
    see Definition~\ref{defn:kleene-free}. We will show that this class of RPQs admits a faster algorithm with complexity $O(|E| + |E|\cdot \sqrt{\out})$.
    \item {\em The RPQ $a^*$:} This RPQ is equivalent to querying for the transitive closure
    of a graph. For this RPQ, we will show that a much simpler and more standard algorithm already meets the same complexity as \ospg. The analysis, however,
    is quite involved.
\end{itemize}

\subsection{Kleene-free RPQs}
The following theorem concerning the output-sensitive complexity of a $k$-path query was recently proved in~\cite{2024arXiv240605536H}:

\begin{theorem}[$k$-path query~\cite{2024arXiv240605536H}]
Given $k$ binary relations $E_1, \ldots, E_k$ for some constant $k$,
consider the following conjunctive query:
\begin{align}
    Q(X_1, X_{k+1}) &\defeq E_1(X_1, X_2) \land E_2(X_2, X_3) \land \ldots \land E_k(X_k, X_{k+1})
\end{align}
The above query can be computed in time\footnote{We made explicit the time needed to read the input by adding the term $N$ to the complexity.} $O(N + N \cdot \sqrt{\out})$, where $N$ is the input size,
 i.e.,~$N\defeq \sum_{i\in[k]} |E_i|$, and $\out$ is the output size, i.e.,~$\out \defeq |Q|$.
\label{thm:k-path}
\end{theorem}

The above theorem immediately implies:
\begin{corollary}[Theorem~\ref{thm:k-path}]
    Let $Q$ be an RPQ of the form $a_1 a_2\ldots a_k$
    where $a_1, \ldots, a_k$ are $k$ not necessarily distinct symbols from the alphabet $\Sigma$.
    Then, $Q$ can be evaluated over any edge-labeled graph $G = (V, E, \Sigma)$ in time $O(|E| + |E| \cdot \sqrt{\out})$.
    \label{cor:k-path}
\end{corollary}

In fact, Theorem~\ref{thm:k-path} further implies the following stronger corollary:
\begin{corollary}[Theorem~\ref{thm:k-path}]
    Let $Q$ be a  Kleene-free RPQ (Definition~\ref{defn:kleene-free}). Then, $Q$ can be evaluated over any edge-labeled graph $G = (V, E, \Sigma)$ in time $O(|E| + |E| \cdot \sqrt{\out})$.
    \label{cor:concatenation-or}
\end{corollary}
\begin{proof}[Proof of Corollary~\ref{cor:concatenation-or}]
By Definition~\ref{defn:kleene-free}, a Kleene-free regular expression $Q$
can only involve concatenation and union operators.
By distributing the union over concatenation,
we can convert $Q$ into a union of a number of regular expressions, each of which is a concatenation of symbols, i.e., has the form $a_1 a_2\ldots a_k$.
We apply Corollary~\ref{cor:k-path} to each of them.
\end{proof}

\subsection{The RPQ $a^*$}
Computing the RPQ $a^*$ is equivalent to computing the {\em transitive closure} of a graph.
Given a directed (but not necessarily edge-labeled) graph $G = (V, E)$, the transitive closure of
$G$ is the list of pairs of vertices $(v, u) \in V \times V$ such that there is a path from $v$ to $u$ in $G$. We use $\out$ to denote the number of pairs in the transitive closure.
We show below that the transitive closure, and by extension the RPQ $a^*$, admit a much
simpler algorithm that already achieves a complexity of 
$O(\min(\out\cdot\sqrt{|E|}, |V|\cdot|E|))$.
Note that this is the same complexity as the general \ospg algorithm from Section~\ref{sec:algo}.
This is because in transitive closure, $|E| \leq \out$, hence the term $|E|^{3/2}$
from Theorem~\ref{thm:complexity} is dominated by the term $\out\cdot\sqrt{|E|}$.

\begin{theorem}[Transitive closure]
    Given any directed graph $G = (V, E)$, the transitive closure of $G$ can be computed in time
    $O(\min(\out\cdot\sqrt{|E|}, |V|\cdot|E|))$.
    \label{thm:transitive-closure}
\end{theorem}

\begin{corollary}[Theorem~\ref{thm:transitive-closure}]
    Given any edge-labeled graph $G = (V, E, \Sigma \defeq \{a\})$, 
    the RPQ $a^*$ can be evaluated over $G$ in time
    $O(\min(\out\cdot\sqrt{|E|}, |V|\cdot|E|))$.
    \label{cor:transitive-closure}
\end{corollary}

Our algorithm for transitive closure that meets the complexity from Theorem~\ref{thm:transitive-closure} is given in Algorithm~\ref{algo:tc}.
This algorithm is basically the standard semi-na\"ive evaluation algorithm for the following Datalog program that defines the transitive closure:
\begin{align}
    T(X, Y)&= E(X, Y)\nonumber\\
    T(X, Y)&= T(X, Z) \land E(Z, Y)\label{eq:linear:tc}
\end{align}
In particular, this semi-na\"ive algorithm iteratively computes a sequence of relations $\delta T^{i}, T^{i}$ for $i = 0, 1, \ldots$ until a fixed point is reached, i.e.~until we encounter $\delta T^i = \emptyset$.
However, the trick lies in the {\em evaluation strategy} for Eq.~\eqref{eq:delta:Ti}.
In particular, we show below that if we evaluate this rule using the right strategy, we can achieve the desired complexity from Theorem~\ref{thm:transitive-closure}.
In Appendix~\ref{app:binary-tc}, we will also consider a different Datalog program for transitive closure and analyze its complexity under semi-na\"ive evaluation.

\begin{algorithm}[th!]
    \caption{Output-sensitive algorithm for transitive closure}
    \begin{algorithmic}[1]
        \Statex{\textbf{Input}}: A directed graph $G=(V,E)$.
        \Statex{\textbf{Output}}: The list of pairs $(v, u) \in V \times V$ where there is a path from $v$ to $u$.
        \Statex{}
        \State{Initialize}
        \begin{align}
            \delta T^0(X, Y) &= E(X, Y)\label{eq:delta:T0}\\
            T^0(X, Y) &= \delta T^0(X, Y)\label{eq:T0}
        \end{align}
        \State{$i \gets 0$}
        \Repeat
        \State{$i \gets i + 1$}
        \State{\change{Evaluate $\delta T^i$ below by iterating over $(x, z) \in \delta T^{i-1}$,
        then over neighbors $y$ of $z$, then checking that $\neg T^{i-1}(x, y)$ holds,
        and if so, adding $(x, y)$ to $\delta T^i$:}}
        \begin{align}
            \delta T^{i}(X, Y) &= \delta T^{i-1}(X, Z) \land E(Z, Y) \land \neg T^{i-1}(X, Y)
            \label{eq:delta:Ti}\\
            T^{i}(X, Y) &= T^{i-1}(X, Y) \lor \delta T^{i}(X, Y)\label{eq:Ti}
        \end{align}
        \Until{$\delta T^{i} = \emptyset$}
        \State \Return $T^i$
    \end{algorithmic}
    \label{algo:tc}
\end{algorithm}

\begin{proof}[Proof of Theorem~\ref{thm:transitive-closure}]
    Consider Algorithm~\ref{algo:tc}.
    Let $I$ be the last iteration, i.e.~the smallest $i$ such that $\delta T^i = \emptyset$.
    Note that $T^I$ is the final output, i.e.~$T^I\defeq T$.
    Moreover notice that by definition, the relations $\delta^0 T, \ldots, \delta^I T$
    form a partition of the output $T$.
    This is because Eq.~\eqref{eq:delta:Ti} ensures that every $\delta^i T$ is disjoint from all the previous $\delta^0 T, \ldots, \delta^{i-1} T$, thanks to the negation $\neg T^{i-1}(X, Y)$.
    \change{For a fixed $i$, we evaluate the rule in Eq.~\eqref{eq:delta:Ti} by going through
    all tuples $(x, z)$ in $\delta^{i-1} T$, and for each one of them, we go over all neighbors $y$ of $z$, then check whether the negation
    $\neg T^{i-1}(x, y)$ holds, and if so, we add $(x, y)$ to $\delta T^i$.
    \footnote{\change{This evaluation strategy is equivalent to using LeapFrogTrieJoin~\cite{LeapFrogTrieJoin2014} or GenericJoin~\cite{SkewStrikesBack2014} to compute the join $\delta T^{i-1}(X, Z) \land E(Z, Y)$ with the variable order $(X, Z, Y)$, and then filtering the result using the negation $\neg T^{i-1}(X, Y)$.}}}
    Therefore, the total time needed to evaluate the rule in Eq.~\eqref{eq:delta:Ti} over all
    iterations $i \in [I]$ is:
    \begin{align}
        \sum_{i\in[I]} \sum_{(x, z)\in \delta T^{i-1}}|\sigma_{Z=z}E(Z, Y)|&\change{=}
        \sum_{(x, z)\in T}
        |\sigma_{Z=z}E(Z, Y)|\label{eq:tc:partition}\\
        &\change{=
        \sum_{(x, z)\in T}
        |\sigma_{Z=z}E(Z, Y)|^{1/2}\cdot |\sigma_{Z=z}E(Z, Y)|^{1/2}}\nonumber\\
        &\leq\sum_{(x, z)\in T}
        |\sigma_{Z=z}E(Z, Y)|^{1/2} \cdot
        |\sigma_{X=x}T(X, Y)|^{1/2}\label{eq:tc:closure}\\
        &=\sum_{\change{(x, z)\in V^2}}
        |\sigma_{X=x, Z=z}T(X, Z)|^{1/2}\cdot
        |\sigma_{Z=z}E(Z, Y)|^{1/2} \cdot
        |\sigma_{X=x}T(X, Y)|^{1/2}\nonumber\\
        &\leq|T(X, Z)|^{1/2}\cdot |E(Z, Y)|^{1/2} \cdot |T(X, Y)|^{1/2}&\label{eq:tc:query-decomposition}\\
        &=\out\cdot\sqrt{|E|}\nonumber
    \end{align}
    \change{Equality~\eqref{eq:tc:partition} holds because the relations $\delta^0 T, \ldots, \delta^I T$ form a {\em partition} of the output $T$, and $\delta^I T=\emptyset$.}
    Inequality~\eqref{eq:tc:closure} holds because by definition of transitive closure,
    for every tuple $(x, z, y)$ where $T(x, z)$ and $E(z, y)$ are true, $T(x, y)$ is also true.
    \change{This implies that for every $(x, z)\in T$, we have 
    $\pi_Y\left(\sigma_{Z=z}E(Z, Y)\right) \subseteq \pi_Y\left(\sigma_{X=x}T(X, Y)\right)$, hence
    $|\sigma_{Z=z}E(Z, Y)| \leq |\sigma_{X=x}T(X, Y)|$.}
    Inequality~\eqref{eq:tc:query-decomposition} holds by the query decomposition lemma, which we re-state below
    for completeness:
    \begin{lemma}[Query Decomposition Lemma~\cite{SkewStrikesBack2014,WCOJGemsOfPODS2018}]
        Let $Q$ be a conjunctive query over a set of variables $\vars(Q)$ and a set of atoms $\atoms(Q)$.
        Given $\bm Y \subseteq \vars(Q)$, let $\left(\lambda_{R(\bm X)}\right)_{R(\bm X)\in\atoms(Q)}$ be a {\em fractional edge cover of} $Q$, i.e.~a set of non-negative weights such that for every $Z\in\vars(Q)$, we have $\sum_{{R(\bm X)\in\, \atoms(Q) \text{ s.t. } Z \in \bm X}} \lambda_{R(\bm X)}\geq 1$.
        Then, the following inequality holds:
        \begin{equation}
            \sum_{\bm y\in \Dom(\bm Y)}
            \underbrace{\prod_{R(\bm X) \in \atoms(Q)}
            |R(\bm X)\ltimes \bm y|^{\lambda_{R(\bm X)}}
            }_{\text{AGM-bound of $Q\ltimes \bm y$}}
            \quad\leq\quad
            \underbrace{\prod_{R(\bm X) \in \atoms(Q)}
            |R|^{\lambda_{R(\bm X)}}}_{\text{AGM-bound of $Q$}}
            \label{eq:query-decom-lemma}
        \end{equation}
        \label{lmm:query-decom-lemma}
    \end{lemma}
    In the above, $\bm y\in \Dom(\bm Y)$ indicates that the tuple $\bm y$ has schema $\bm Y$. Moreover,
    $R(\bm X)\ltimes \bm y$ denotes the {\em semijoin} of the atom
    $R(\bm X)$ with the tuple $\bm y$.

    In particular, in Inequality~\eqref{eq:tc:query-decomposition}, we apply the query decomposition lemma to the query $T(X, Z) \land E(Z, Y) \land T(X, Y)$ with the fractional edge cover $\left(\frac{1}{2}, \frac{1}{2}, \frac{1}{2}\right)$.
    Inequality~\eqref{eq:tc:query-decomposition} bounds the total time needed to evaluate the rule in Eq.~\eqref{eq:delta:Ti} over all $i \in [I]$ by $O(\out\cdot\sqrt{|E|})$.
    In addition, we can also bound the same time by $O(|V|\cdot|E|)$ as follows:
    \begin{align*}
        \sum_{i\in[I]} \sum_{(x, z)\in \delta T^{i-1}}|\sigma_{Z=z}E(Z, Y)|&\change{=}
        \sum_{(x, z)\in T}
        |\sigma_{Z=z}E(Z, Y)|\\
        &\leq\sum_{x\in \pi_XT} \sum_{\change{z\in V}} |\sigma_{Z=z}E(Z, Y)|\\
        &\leq\sum_{x\in \pi_XT} |E|\\
        &\leq |V|\cdot |E|
    \end{align*}

    \change{Eq.~\eqref{eq:delta:T0}} can be evaluated in time $O(|E|)$, which is $O(\out)$ since the output $T$ of transitive closure is a superset of the input graph.
    The evaluation time for Eq.~\eqref{eq:T0} and~\eqref{eq:Ti}
    is dominated by the evaluation time for Eq.~\eqref{eq:delta:T0} and~\eqref{eq:delta:Ti}.
\end{proof}

\section{Related Work}
\label{sec:related}

\paragraph{RPQ Evaluation}
Recent work gives an excellent treatment of the RPQ evaluation problem investigated in this paper~\cite{CaselS23}: It overviews the \pg algorithm and gives conditional lower bounds for the problem (discussed in the introduction).

There is extensive literature on the asymptotic complexity of variants of the decision problem associated with the RPQ evaluation problem that we study in this paper: Given an edge-labeled directed data graph, an RPQ $Q$ and two vertices $v$ and $u$ in the data graph, decide whether there is a path from $v$ to $u$ that matches $Q$.
This problem can be solved in polynomial time combined complexity, i.e., when the sizes of both the graph and the query are part of the input~\cite{BarceloLLW12}. 
The problem has NL-complete data complexity~\cite{Baeza13}.
The enumeration of arbitrary or shortest paths between $v$ and $u$ that match the given query admits  polynomial delay in data complexity~\cite{MartensT19}. In case the queries are expressed as SPARQL property path expressions under the W3C semantics, the problem is NP-complete even under data complexity
~\cite{LosemannM13, MendelzonW95}.

Two widely studied variants of the decision problem associated with the RPQ evaluation  problem consider whether the given vertices are connected by a {\em simple path}, i.e., a path in which  each vertex occurs at most once~\cite{MendelzonW95, MartensT19, BaganBG20}
or by a {\em trail}, i.e., a path in which every edge occurs at most once~\cite{MartensT19, MartensNP23}.  
In case of simple paths, the problem is NP-complete even for fixed basic queries such as $(aa)^*$ or $a^*ba^*$~\cite{MendelzonW95}. 
Depending on the structure of the regular expression 
and the corresponding automaton, 
the data complexity of the problem is either in  $AC^0$, NL-complete, or NP-complete~\cite{BaganBG20}. 
A similar classification is established in case of trails, where the  tractable class is shown to be larger than in case of simple paths~\cite{MartensNP23}. 
For both the simple path and the trail semantics, the literature characterizes fragments of regular path queries for which the problem is 
fixed-parameter tractable, with the query size as the  parameter~\cite{MartensT19}.

%identifies the class of so-called {\em simple transitive expressions}, which are argued to be common in practical applications, and characterizes fragments 
%for which the decision problem is 
%fixed-parameter tractable (query size is the parameter) under the simple path or the trail semantics~\cite{MartensT19}.

A common extension of RPQs is given by {\em conjunctive} regular path queries, which is a conjunctive query over binary atoms defined by regular path queries. The combined and data complexity of the above decision problem becomes NP-complete and respectively remains NL-complete for conjunctive RPQs~\cite{BarceloLLW12}. The complexities do not increase when we consider the so-called {\em injective} semantics, which naturally generalizes the simple-path semantics for the class of conjunctive RPQs~\cite{FigueiraR23}.
The worst-case optimality of conjunctive regular path queries has been recently investigated~\cite{DBLP:conf/icdt/CucumidesRV23}, where the cardinalities of the regular path queries present in a conjunctive regular path query are assumed to be known.
Extensions of conjunctive regular path queries were also investigated: capture variables~\cite{Schmid22}, 
output paths and relationships among them~\cite{BarceloLLW12},
and combined data and topology querying on graph databases~\cite{LibkinMV16, FreydenbergerS13}.
Further work studies the complexity of checking containment for conjunctive regular path queries   \cite{CalvaneseGLV00, ReutterRV17, GluchMO19, BarceloF019, FigueiraGKMNT20, Figueira20, FigueiraM23} and of measuring the  contribution of edges and vertices to query  answers via the Shapley value~\cite{KhalilK23}.

To represent the set of possible paths matching a regular path query, a succinct and lossless representation has been introduced~\cite{MartensNPRVV23} and recently connected~\cite{DBLP:conf/icdt/KimelfeldMN25} to factorized representations~\cite{DBLP:journals/tods/OlteanuZ15} of the output of conjunctive queries.
The representation can be computed in linear time combined complexity and allows to perform 
operations like enumeration, counting,
random sampling, grouping, and taking unions, which are common operations 
on the tabular representation of all paths.
An in-depth analytical study has been made on SPARQL queries formulated by end users on top of graph-structured data in order to explain 
why regular path queries in graph database applications behave better than worst-case complexity results suggest~\cite{BonifatiMT20, MartensT19b}.

In knowledge representation, work examines the complexity of two-way conjunctive regular path queries over knowledge bases expressed
by means of linear existential rules~\cite{BienvenuT16}, lightweight descriptive logic~\cite{BienvenuOS15}, or guarded existential rules~\cite{BagetBMT17}.

%%%%%%%%%%%%%%%%%%%%
% Output-sensitive
%%%%%%%%%%%%%%%%%%%%
\paragraph{Output-sensitive query evaluation.}
Our output-sensitive query evaluation algorithm \ospg follows a well-established line of work on output-sensitive algorithms, which we overview next.

A growing line of work decomposes the query processing task into a pre-processing phase, where a data structure representing succinctly the query output is constructed, and an enumeration phase, where each tuple in the query output is enumerated with some delay $d$. The overall data complexity is expressed as $O(N^w + d\cdot \out)$, where $N$ is the size of the input database, $\out$ is the output size, and $w$ is a width parameter that depends on the query structure. A large number of results can be explained in this complexity framework, as exemplified next.

We first consider the case of $\alpha$-acyclic conjunctive queries. The classical Yannakakis algorithm, developed more than four decades ago, is the first output-sensitive evaluation algorithm for acyclic queries: It runs in time  $O(N + N\cdot \out)$ data complexity~\cite{Yannakakis81}, i.e., for $w=1$ and $d=O(N)$~\cite{BaganDG07}. For free-connex $\alpha$-acyclic conjunctive queries, $d$ decreases to $O(1)$ while $w$ remains 1~\cite{BaganDG07}. For hierarchical queries, which are a strict subset of the class of $\alpha$-acyclic conjunctive queries,  $w = 1+ (\omega -1)\epsilon$ and $d = O(N^{1 - \epsilon})$, where $\omega$ is the fractional hypertree width of the query (adapted from Boolean to conjunctive queries with arbitrary free variables) and $\epsilon\in[0,1]$ can be chosen arbitrarily~\cite{KaraNOZ23}. There are similar trade-offs between $w$ and $d$ for general $\alpha$-acyclic queries~\cite{KaraNOZ23_acyclic} and queries with access patterns~\cite{DBLP:conf/icdt/00020OZ23, DBLP:conf/pods/ZhaoDK23}. There are recent polynomial-time factor improvements on the Yannakakis algorithm~\cite{DBLP:conf/sigmod/DeepHK20,DBLP:journals/pacmmod/Hu24,2024arXiv240605536H,DBLP:journals/pacmmod/DeepZFK24}, following an initial observation by Amossen and Pagh~\cite{DBLP:conf/icdt/AmossenP09} on the output-sensitive evaluation of a two-relation join-project query that encodes the multiplication of two sparse matrices.
Any $\alpha$-acyclic conjunctive query can be evaluated  in time $O(N + \out + N\cdot\out^{5/6})$ using a non-combinatorial algorithm (using fast matrix multiplication)~\cite{DBLP:journals/pacmmod/Hu24} and in time $O(N + \out + N\cdot\out^{1-\epsilon})$ using a combinatorial algorithm~\cite{DBLP:journals/pacmmod/DeepZFK24}, where $\epsilon\in (0,1]$ is a query-dependent constant. These complexities also fit the general complexity framework by taking $w=1$ and $d=O(1+ N/\out^{1/6})$ and respectively $d=O(1+ N/\out^{1/k})$. The $k$-path query discussed in Section~\ref{sec:special} can be evaluated in time $O(N + N\cdot \out^{1/2})$~\cite{2024arXiv240605536H} or in time $O(N + \out + N\cdot\out^{1-1/k})$~\cite{DBLP:journals/pacmmod/DeepZFK24}. A lower bound of $\Omega(N + N\cdot\out^{1/2})$ is also known for this query~\cite{2024arXiv240605536H}. This line of work culminates with an output-sensitive semi-ring extension of the Yannakakis algorithm to achieve $O(N+\out + N\cdot\out^{1- 1/\textsf{outw}})$ data complexity\footnote{We ignore here a $\polylog(N)$ factor.} for any acyclic join-aggregate query, where \textsf{outw} is  the so-called out-width of the query~\cite{2024arXiv240605536H}; this complexity fits the  complexity framework with $w=1$ and $d=O(1+N/\out^{1/\textsf{outw}})$. This upper bound is accompanied by a matching lower bound (modulo a $\polylog(N)$ factor). 

Arbitrary conjunctive queries (so including cyclic queries) can be evaluated using the factorized databases framework~\cite{DBLP:journals/tods/OlteanuZ15} where the width $w$ is the fractional hypertree width\footnote{This  is generalized from Boolean to conjunctive queries; $w$ is denoted as $s^\uparrow$ in~\cite{DBLP:journals/tods/OlteanuZ15} and FAQ-width in~\cite{faq}.} and the delay $d=O(1)$. For conjunctive queries over the Boolean semiring, $w$ becomes the submodular width, which is smaller than the fractional hypertree width for many queries, and $d=O(\polylog(N))$ using the PANDA algorithm~\cite{panda:pods17, BerkholzS19}. For conjunctive queries over arbitrary semirings, $w$ becomes the sharp submodular width, which is sandwiched between the 
fractional hypertree width and the submodular width, and $d$ remains $O(\polylog(N))$ using the FAQ-AI framework~\cite{DBLP:journals/tods/KhamisCMNNOS20}.

A further line of work observes constraints on the enumeration order of the tuples in the query output. For free-connex $\alpha$-acyclic queries, the query output can be enumerated in random order with $w = 1$ and $d = O(\polylog(N))$~\cite{CarmeliZBCKS22}. An extensive line of work investigates 
the efficient ranked enumeration for full conjunctive queries~\cite{TziavelisAGRY20}, arbitrary conjunctive queries~\cite{DeepK21, CarmeliTGKR23}, theta-joins~\cite{TziavelisGR21}, and monadic second order 
(MSO) logic~\cite{DBLP:conf/icdt/BourhisGJR21, DBLP:conf/icdt/AmarilliBCM24}.
MSO queries over words can be evaluated with 
$w=1$ and $d= O(\log(N))$ such that the result is enumerated 
following scores assigned to output words defined using so-called
MSO cost functions. This result is generalized to MSO 
queries over trees~\cite{DBLP:conf/icdt/AmarilliBCM24}.
A subclass of free-connex $\alpha$-acyclic queries, called queries 
without a {\em disruptive trio}\footnote{The negation of the notion of "disruptive trio" is equivalent to an earlier characterization of so-called free-top variable orders that allow for efficient enumeration in a given order in the context of factorized databases~\cite{DBLP:journals/pvldb/BakibayevKOZ13,KaraNOZ23}.}, admit ranked enumeration following a lexicographic order using $w=1$ (modulo a $\log(N)$ factor) and $d = O(\log(N))$~\cite{CarmeliTGKR23}. Using the same parameters $w$ and $d$, ranked enumeration following any ranking function that can be interpreted as a selective dioids can be supported for full $\alpha$-acyclic conjunctive queries with inequality conditions~\cite{TziavelisGR21}.

Output-sensitive algorithms have also been studied in the context of the Massively Parallel Computation (MPC) model~\cite{BeameKS14,HuYT19}, including 
an MPC adaptation of the Yannakakis algorithm~\cite{AfratiJRSU17}.
The literature provides an almost complete characterization of 
the $\alpha$-acyclic joins with respect to instance optimality and output optimality
in the MPC model~\cite{Hu019}: Instance optimality can be obtained for a strict subclass of $\alpha$-acyclic queries called r-hierarchical~\cite{Hu019}.

\paragraph{\change{Output-sensitive transitive closure.}}\change{
    Let $G=(V, E)$ be a directed graph.
    Using fast matrix multiplication,
    it is possible to compute the transitive closure of $G$ in time
    $\tilde O(|V|^{\omega})$~\cite{kozen92}, where $\omega$ is the {\em matrix multiplication exponent}, i.e.
    the smallest exponent $\alpha$ that allows us to multiply two $N\times N$ matrices
    in time $\tilde O(N^\alpha)$, and $\tilde O$ hides a factor of $\polylog(N)$.
    Currently, the best known upper bound for $\omega$ is $2.371552$~\cite{DBLP:conf/soda/WilliamsXXZ24}.
    An output-sensitive extension of this result computes the transitive closure of $G$ in time $\tilde O(|V|^{\frac{\omega+1}{4}}\cdot\out)$ \cite{BORASSI201651},
    which translates to $\tilde O(|V|^{0.842888}\cdot \out)$
    using the best known $\omega$, where $\out$ is the number of edges in the transitive closure.
    Without using fast matrix multiplication,~\cite{10.1145/2745754.2745779} presents a combinatorial algorithm for transitive closure that runs in time $O(\out^{3/2})$.
    Using fast matrix multiplication,~\cite{doi:10.1137/1.9781611977912.167} improves this
    bound to $\tilde O(\out^{1.3459})$, under the best known $\omega$.
    Note that in transitive closure, the input graph is always a subset of the output,
    hence $|E| \leq \out$.
    Therefore, our $\ospg$ transitive closure bound, $O(\out \cdot\sqrt{|E|})$,
    is never larger than the bound of $O(\out^{3/2})$ from~\cite{10.1145/2745754.2745779},
    and is strictly smaller whenever $\out > |E|$.
    Moreover, our $\ospg$ bound is strictly smaller than the bound of $O(\out^{1.3459})$ from~\cite{doi:10.1137/1.9781611977912.167} whenever $\out > |E|^{1.4455}$.
}

\section{Conclusion}
\label{sec:conclusion}

In this paper, we introduced a new output-sensitive algorithm for the evaluation of regular path queries over graphs. We showed that the data complexity of our algorithm is upper bounded by that of the product graph algorithm, which is the best-known prior algorithm and the workhorse of practical implementations in graph database management systems, and even asymptotically improves on the latter for sparse data graphs and selective queries.

One line of future work is to extend the investigation of output-sensitive algorithms to the evaluation of conjunctive regular path queries, which use regular expressions to connect variables in the query. Furthermore, there is a notable lack of output-sensitive complexity lower bounds. This requires  output-sensitive refinements of existing conjectures, in the spirit of the sparse Boolean matrix multiplication conjecture mentioned in footnote (5) in the introduction. 

We would also like to understand whether the asymptotic complexity gap between our output-sensitive algorithm and the classical product graph approach also translates into a runtime performance gap for practical workloads. 

%%%%%%%%%%%%%%%%%%%%%%%%%%%%%%%%%%%%%%%%%%%%%%%%%%%%%%%%%%%%%%%%%%%%%%%%%%%%%%%%%%%%%%%%%%%%
\begin{acks}
    This work was partially supported by NSF-BSF 2109922, NSF-IIS 2314527, NSF-SHF 2312195,
    and SNSF 200021-231956.
\end{acks}
%%%%%%%%%%%%%%%%%%%%%%%%%%%%%%%%%%%%%%%%%%%%%%%%%%%%%%%%%%%%%%%%%%%%%%%%%%%%%%%%%%%%%%%%%%%%

\bibliographystyle{ACM-Reference-Format}
\bibliography{bibliography}

% \pagebreak

\appendix
\section{Transitive Closure: Linear versus Binary Formulation}
\label{app:binary-tc}
In the proof of Theorem~\ref{thm:transitive-closure}, we analyzed the overall runtime needed for semi-na\"ive evaluation of the {\em linear}
formulation of transitive closure given by Eq.~\eqref{eq:linear:tc}.
By {\em linear} here, we mean that the recursive rule given by Eq.~\eqref{eq:linear:tc} has a single
recursive atom $T(X, Z)$ in its body.
In contrast, the transitive closure can also be defined using a {\em binary} formulation, where the recursive rule has two recursive atoms in its body:
\begin{align}
    T(X, Y)&= E(X, Y)\nonumber\\
    T(X, Y)&= T(X, Z) \land T(Z, Y)\label{eq:binary:tc}
\end{align}
It is well-known that given a graph $G = (V, E)$ of diameter $d$, the linear formulation of transitive closure
takes $O(d)$ iterations to converge, while the binary formulation takes $O(\log d)$ iterations.
It might be tempting to conclude that the binary formulation is faster overall than the linear formulation,
but we show below that this is not the case. In particular, while we already showed
in the proof of Theorem~\ref{thm:transitive-closure} that the linear formulation
takes time $O(\min(\out \cdot\sqrt{|E|}, |V|\cdot |E|))$, we show below that the binary formulation takes time $O(\out^{3/2})$. Since the transitive closure of a graph $G$ contains $G$, we have $\out \geq |E|$
and $\out$ could be as large as $|E|^2$.

In order to show that the binary formulation takes time $O(\out^{3/2})$, we need to analyze the runtime of the semi-na\"ive evaluation algorithm for the binary formulation, given by Eq.~\eqref{eq:binary:tc},
similar to what we did in the proof of Theorem~\ref{thm:transitive-closure}.
Semi-na\"ive evaluation of the program given by Eq.~\eqref{eq:binary:tc} evaluates
a sequence of relations $\delta T^i, T^i$ for $i = 0, 1, \ldots$ until a fixpoint is reached,
i.e., until we encounter $\delta T^i = \emptyset$:
\begin{align}
    \delta T^0(X, Y) &= E(X, Y)\label{eq:delta:T0:binary}\\
    T^0(X, Y) &= \delta T^0(X, Y)\label{eq:T0:binary}\\
    \text{for $i = 1, 2, \ldots$}\nonumber\\
    \delta T^{i}(X, Y) &= \delta T^{i-1}(X, Z) \land T^{i-1}(Z, Y) \land \neg T^{i-1}(X, Y)
    \label{eq:delta:Ti:binary1}\\
    \delta T^{i}(X, Y) &= T^{i-1}(X, Z) \land \delta T^{i-1}(Z, Y) \land \neg T^{i-1}(X, Y)
    \label{eq:delta:Ti:binary2}\\
    T^{i}(X, Y) &= T^{i-1}(X, Y) \lor \delta T^{i}(X, Y)\label{eq:Ti:binary}
\end{align}
\change{Let $I$ be the last iteration, i.e.~the smallest $i$ such that $\delta T^i = \emptyset$.}
We bound the total time needed to evaluate the rule in Eq.~\eqref{eq:delta:Ti:binary1} over all
iterations. Rule~\eqref{eq:delta:Ti:binary2} is similar.
Similar to the proof of Theorem~\ref{thm:transitive-closure}, the relations $\delta T^i$
form a partition of the final $T$, thanks to the negation
$\neg T^{i-1}(X, Y)$ in Eq.~\eqref{eq:delta:Ti:binary1} and~\eqref{eq:delta:Ti:binary2}. Hence, the total time needed to evaluate Eq.~\eqref{eq:delta:Ti:binary1} over all iterations is upper bounded by:
\begin{align}
    \sum_{i\in[I]} \sum_{(x, z)\in \delta T^{i-1}}|\sigma_{Z=z}T(Z, Y)|&\change{=}
    \sum_{(x, z)\in T}
    |\sigma_{Z=z}T(Z, Y)|\nonumber\\
    &\change{=\sum_{(x, z)\in T}
    |\sigma_{Z=z}T(Z, Y)|^{1/2} \cdot
    |\sigma_{Z=z}T(Z, Y)|^{1/2}}\nonumber\\
    &\leq\sum_{(x, z)\in T}
    |\sigma_{Z=z}T(Z, Y)|^{1/2} \cdot
    |\sigma_{X=x}T(X, Y)|^{1/2}\label{eq:tc:binary:def}\\
    &=\sum_{\change{(x, z)\in V^2}}
    |\sigma_{X=x, Z=z}T(X, Z)|^{1/2}\cdot
    |\sigma_{Z=z}T(Z, Y)|^{1/2} \cdot
    |\sigma_{X=x}T(X, Y)|^{1/2}\nonumber\\
    &\leq|T(X, Z)|^{1/2}\cdot |T(Z, Y)|^{1/2} \cdot |T(X, Y)|^{1/2}\nonumber\\
    &=\out^{3/2}\nonumber
\end{align}
Inequality~\eqref{eq:tc:binary:def} holds because by definition of transitive closure,
for every tuple $(x, z, y)$ where $T(x, z)$ and $T(z, y)$ are true, $T(x, y)$ must be true as well.
This proves that the total time needed to evaluate the rule in Eq.~\eqref{eq:delta:Ti:binary1}, and by symmetry Eq.~\eqref{eq:delta:Ti:binary2}, over all iterations is $O(\out^{3/2})$.
Rule~\eqref{eq:delta:T0:binary} takes time $O(|E|) = O(\out)$, and rules~\eqref{eq:T0:binary} and~\eqref{eq:Ti:binary} are dominated by other rules. This proves that the overall time for semi-na\"ive
evaluation of the binary formulation of transitive closure is $O(\out^{3/2})$.

\end{document}